\documentclass[12pt]{article}
\usepackage{amsmath}
\usepackage[pdftex]{graphicx}
\usepackage{amsthm}
\usepackage{amssymb}
\usepackage{fancyhdr}
\usepackage{rotating}
\usepackage{epstopdf}
\DeclareGraphicsExtensions{.eps}
\newtheorem{theorem}{Theorem}
\usepackage{hyperref}

\begin{document}

\begin{center}
\section*{Variance and Volatility Swaps and Futures Pricing for Stochastic Volatility Models}
{\sc Anatoliy Swishchuk}\\Department of Mathematics and Statistics\\University of Calgary\\2500 University Drive NW\\Calgary, Alberta, Canada, T2N 1N4\\

\vspace{0.5cm}

{\sc Zijia Wang}\\Department of Mathematics and Statistics\\University of Calgary\\2500 University Drive NW\\Calgary, Alberta, Canada, T2N 1N4\\

\end{center}

\hspace{1cm}

{\bf Abstract:} In this chapter, we consider volatility swap, variance swap and VIX future pricing under different stochastic volatility models and jump diffusion models which are commonly used in financial market. We use convexity correction approximation technique and Laplace transform method to evaluate volatility strikes and estimate VIX future prices. In empirical study, we use Markov chain Monte Carlo algorithm for model calibration based on S\&P 500 historical data, evaluate the effect of adding jumps into asset price processes on volatility derivatives pricing, and compare the performance of different pricing approaches.

\hspace{1cm}

{\bf Keywords}: variance swap, volatility swap, stochastic volatility, VIX future, convexity correction, Markov chain Monte Carlo

\section{Variance and Volatility Swaps for Stochastic Volatility Models}

In this section, we will focus on the variance and volatility swap pricing under stochastic volatility models and stochastic volatility models with jumps. The continuous variance strike under these models can be found through definition. However, the non-linearity property of square root function requires us to apply some techniques when evaluating the continuous volatility strike.  In the following sections, we will use the convexity correction formula to approximate volatility strikes, and the closed-form solutions developed in [Broadie and Jain, 2008] will also be presented, for the sake of completeness of the presentation.

\subsection{Heston Stochastic Volatility Model}
We assume all the price dynamics are modelled under risk neutral measure. Now we present an analysis of variance and volatility swaps under Heston stochastic volatility model. The Heston model [1993] is given by 
\begin{align}
\nonumber&dS_t=r S_tdt+\sqrt{V_t}S_tdW_t^1 \\ 
&dV_t=\kappa(\theta-V_t)dt+\sigma\sqrt{V_t}dW_t^2, \label{h1}
\end{align} 
the first equation in (\ref{h1}) gives the dynamics of the stock price $S_t$, $r$ is the spot interest rate, $\sqrt{V_t}$ is the volatility of stock price and the variance $V_t$ is a C-I-R process. $\kappa$ represents the speed of mean reversion, $\theta$ is the long run average of variance and $\sigma$ is the volatility of variance. $W_t^1$ and $W_t^2$ are two standard Brownian motion with correlation $\rho$.\

Instead of finding the explicit solution for Heston model, we will derive the mean and variance of $V_t$. To do so, we let 
\[
V_t=e^{-\kappa t}Z_t~ \text{with} ~Z_0=V_0,
\]
then
\begin{equation*}
       dV_t=-\kappa e^{-\kappa t}Z_tdt+e^{-\kappa t}dZ_t=-\kappa V_tdt+e^{-\kappa t}dZ_t,
\end{equation*}
and 
\begin{equation*}
       dZ_t=\kappa\theta e^{\kappa t}dt+\sigma e^{\kappa t}\sqrt{V_t}dW_t^2,
\end{equation*}
Take the integration and substitute initial value we have that
\begin{equation*}
       Z_t=\theta( e^{\kappa t}-1)+\sigma\int_0^t e^{\kappa s}\sqrt{V_s}dW_s^2+V_0
\end{equation*}
and therefore
\begin{equation}
    V_t=e^{-\kappa t}Z_t=\theta+(V_0-\theta) e^{-\kappa t}+\sigma e^{-\kappa t}\int_0^t e^{\kappa s}\sqrt{V_s}dW_s^2. \label{h0}
\end{equation}
Notice that expectation of It\^{o} integral is zero, we have that
\begin{equation}
    E(V_t)=e^{-\kappa t}Z_t=\theta+(V_0-\theta) e^{-\kappa t}.
\end{equation}
By using It\'{o} isometry property we obtain following formula for variance of $V_t$
\begin{align}
Var(V_t)&=\sigma^2 e^{-2\kappa t}Var(\int_0^t e^{\kappa s}\sqrt{V_s}dW_s^2)\nonumber\\
&=\sigma^2 e^{-2\kappa t}E(\int_0^t e^{2\kappa s}V_sds)\nonumber\\
&=\sigma^2 e^{-2\kappa t}\int_0^te^{2\kappa s}(\theta+(V_0-\theta) e^{-\kappa s})ds\nonumber\\
&=\frac{\theta\sigma^2}{2\kappa}(1-e^{-2\kappa t})+\frac{(V_0-\theta)\sigma^2}{\kappa}(e^{-\kappa t}-e^{-2\kappa t}).
\end{align}\

\subsubsection{Variance Swap for the Heston's Model}
In the case of Heston stochastic volatility model, continuous realized variance is given by
\begin{equation}
        RV_c(0,T)=\frac{1}{T}\int_0^TV_tdt \label{h11}
\end{equation}
and the fair continuous variance strike is given by
\begin{equation}
    K_{var}^*=E\big(\frac{1}{T}\int_0^TV_tdt\big)=\theta+\frac{V_0-\theta}{\kappa T}(1-e^{-\kappa T}).\label{h15}
\end{equation}

\subsubsection{Volatility Swap for the Heston's Model (Convexity Correction Method)}
The realized volatility is commonly calculated by using the square root of the realized variance define in (\ref{h11}), and the fair continuous volatility strike $K_{vol}^*$ for Heston's model is given by 
\begin{equation}
    K_{vol}^*=E\big(\sqrt{RV_c(0,T)}\big)=E\big(\sqrt{\frac{1}{T}\int_0^TV_tdt}\big).
\end{equation}

For fair continuous volatility strike $K_{vol}^*$, we have following theorem.
\begin{theorem}
The fair continuous volatility strike under Heston's model can be approximated as following
\begin{align}
    K_{vol}^*&\approx \sqrt{\theta+\frac{V_0-\theta}{\kappa T}(1-e^{-\kappa T})}-\frac{\frac{\sigma^2e^{-2\kappa T}}{2\kappa^3T^2}\big[(V_0-\theta)(2e^{2\kappa T}-4e^{\kappa T}\kappa T-2)\big]}{8\big(\theta+\frac{V_0-\theta}{\kappa T}(1-e^{-\kappa T})\big)^{\frac{3}{2}}}\nonumber\\
    &-\frac{\frac{\sigma^2e^{-2\kappa T}}{2\kappa^3T^2}\big[\theta(2e^{\kappa T}\kappa T-3e^{2\kappa T}+4e^{\kappa T}-1)\big]}{8\big(\theta+\frac{V_0-\theta}{\kappa T}(1-e^{-\kappa T})\big)^{\frac{3}{2}}}
\end{align}
\end{theorem}

\begin{proof}

To evaluate the fair discrete volatility strike under Heston's model, we need to find the risk neutral expectation of the square root of realized variance. Brockhaus and Long [2002] show that the fair volatility strike $K_{vol}^*$ can be approximated by using Taylor's expansion of $\sqrt{RV_c(0,T)}$ around $E\big(RV_c(0,T)\big)$ as following
\begin{align}
    \sqrt{RV_c(0,T)}&\approx \sqrt{E\big(RV_c(0,T)\big)}+\frac{RV_c(0,T)-E\big(RV_c(0,T)\big)}{2\sqrt{E\big(RV_c(0,T)\big)}}\nonumber\\
    &-\frac{\bigg(RV_c(0,T)-E\big(RV_c(0,T)\big)\bigg)^2}{8\sqrt{\bigg(E\big(RV_c(0,T)\big)\bigg)^3}}.\label{ag1}
\end{align}
Taking expectations under the risk-neutral measure on both sides of (\ref{ag1}) gives
\begin{align}
    K_{vol}^*=E(\sqrt{RV_c(0,T)})\approx \sqrt{E\big(RV_c(0,T)\big)}-\frac{Var\big(RV_c(0,T)\big)}{8\big(E(RV_c(0,T))\big)^{\frac{3}{2}}}.\label{ag2}
\end{align}
Thus, the fair volatility strike can be approximated by the convexity correction formula (\ref{ag2}), which only requires to find expectation and variance of realized variance. \

Specifically, for the Heston's model, the expectation of realized variance is given by (\ref{h15}) while the variance can be found through following steps.\

Since
\begin{align}
&~~~~~E\big(\int_0^t e^{\kappa x}\sqrt{V_x} dW_x^2\cdot\int_0^s e^{\kappa x}\sqrt{V_x} dW_x^2\big)\nonumber\\
&=\int_0^{t\wedge s}E(e^{2\kappa x}V_x)dx\nonumber\\
&=\frac{\theta}{2\kappa}e^{2\kappa(t\wedge s)}+\frac{V_0-\theta}{\kappa}e^{\kappa (t\wedge s)}-\frac{\theta}{2\kappa}-\frac{V_0-\theta}{\kappa},\label{h17}
\end{align}
and from (\ref{h0}) we have that
\begin{align}
E(V_tV_s)=&(V_0^2+\theta^2)e^{-\kappa(t+s)}-2V_0\theta e^{-\kappa (t+s)}+(V_0\theta-\theta^2)(e^{-\kappa t}+e^{-\kappa s})+\theta^2\nonumber\\
&+\sigma^2e^{-\kappa(t+s)}E\big(\int_0^t e^{\kappa x}\sqrt{V_x} dW_x^2\cdot\int_0^s e^{\kappa x}\sqrt{V_x} dW_x^2\big)\nonumber\\
=&(V_0-\theta)^2e^{-\kappa(t+s)}+(V_0\theta-\theta^2)(e^{-\kappa t}+e^{-\kappa s})+\theta^2+\sigma^2e^{-\kappa(t+s)}\frac{\theta}{2\kappa}e^{2\kappa(t\wedge s)}\nonumber\\
&+\sigma^2e^{-\kappa(t+s)}\frac{V_0-\theta}{\kappa}e^{\kappa (t\wedge s)}+\frac{\sigma^2e^{-\kappa(t+s)}(\theta-2V_0)}{2\kappa},\label{h19}
\end{align}
and,
\begin{align}
Var\big(\frac{1}{T}\int_0^TV_tdt\big)&=E\bigg(\big(\frac{1}{T}\int_0^TV_tdt\big)^2\bigg)-E^2\big(\frac{1}{T}\int_0^TV_tdt\big)\nonumber\\
&=\frac{1}{T^2}\int_0^T\int_0^TE(V_tV_s)dsdt-E^2\big(\frac{1}{T}\int_0^TV_tdt\big).\label{h21}
\end{align}

substitutes (\ref{h15}) and (\ref{h19}) into (\ref{h21}) we have
\begin{align}
Var\big(\frac{1}{T}\int_0^TV_tdt\big)&=\frac{\sigma^2e^{-2\kappa T}}{2\kappa^3T^2}\big[(V_0-\theta)(2e^{2\kappa T}-4e^{\kappa T}\kappa T-2)\nonumber\\
&+\theta(2e^{\kappa T}\kappa T-3e^{2\kappa T}+4e^{\kappa T}-1)\big].\label{h23}
\end{align}

From the convexity correction formula (\ref{ag2}) we have 
\begin{align}
    K_{vol}^*&\approx \sqrt{E\big(RV_c(0,T)\big)}-\frac{Var\big(RV_c(0,T)\big)}{8\big(E(RV_c(0,T))\big)^{\frac{3}{2}}}\nonumber\\
    &=\sqrt{\theta+\frac{V_0-\theta}{\kappa T}(1-e^{-\kappa T})}-\frac{\frac{\sigma^2e^{-2\kappa T}}{2\kappa^3T^2}\big[(V_0-\theta)(2e^{2\kappa T}-4e^{\kappa T}\kappa T-2)\big]}{8\big(\theta+\frac{V_0-\theta}{\kappa T}(1-e^{-\kappa T})\big)^{\frac{3}{2}}}\nonumber\\
    &-\frac{\frac{\sigma^2e^{-2\kappa T}}{2\kappa^3T^2}\big[\theta(2e^{\kappa T}\kappa T-3e^{2\kappa T}+4e^{\kappa T}-1)\big]}{8\big(\theta+\frac{V_0-\theta}{\kappa T}(1-e^{-\kappa T})\big)^{\frac{3}{2}}}.\label{h33}
\end{align}

\end{proof}

\subsubsection{Volatility Swap for the Heston's Model (Laplace Transform Method)}
It is convenient to use convexity correction formula to approximate the fair volatility strike, but the realized variance is required to be in the radius of convergence to make the first three terms in the Taylor expansion be a good approximation of square root function. Broadie and Jain [2008] claim that the 4th order terms in Taylor expansion of square root function are not small enough in the Heston stochastic volatility model, and hence the convexity correction formula will not provide a good estimate of the fair volatility strike. Instead of using Taylor expansion, they consider following formula for square root function
\begin{equation}
    \sqrt{x}=\frac{1}{2\sqrt{\pi}}\int_0^\infty\frac{1-e^{-sx}}{s^{\frac{3}{2}}}ds, \label{lp1}
\end{equation}
and by taking expectation on both sides of (\ref{lp1}) and using Fubini's theorem we have that
\begin{equation}
    E(\sqrt{x})=\frac{1}{2\sqrt{\pi}}\int_0^\infty\frac{1-E(e^{-sx})}{s^{\frac{3}{2}}}ds. \label{lp2}
\end{equation}

\begin{theorem}
 For the Heston stochastic volatility model, the fair continuous volatility strike is given by 
\begin{align}
     K_{vol}^*=E(\sqrt{RV_c(0,T)})=\frac{1}{2\sqrt{\pi}}\int_0^\infty\frac{1-E(e^{-sRV_c(0,T)})}{s^{\frac{3}{2}}}ds, 
\end{align}
where
\begin{equation}
    E(e^{-sRV_c(0,T)})=\exp\big(A(T,s)-B(T,s)V_0\big), \label{ht1}
\end{equation}
\begin{align}
     &A(T,s)=\frac{2\kappa\theta}{\sigma^2}\log\bigg(\frac{2\gamma(s)e^{\frac{T(\gamma(s)+\kappa)}{2}}}{(\gamma(s)+\kappa)(e^{T\gamma(s)}-1)+2\gamma(s)}\bigg),\nonumber\\
     &B(T,s)=\frac{2s(e^{T\gamma(s)}-1)}{T(\gamma(s)+\kappa)(e^{T\gamma(s)}-1)+2T\gamma(s)},\nonumber\\
     &\gamma(s)=\sqrt{\kappa^2+\frac{2\sigma^2s}{T}}.\nonumber
\end{align}
\end{theorem}
See [Broadie and Jain, 2008, Prop. 3.1, page 774] for more details.  The above formula for the Laplace transform of the continuous realized variance can be justified by using Feynman-Kac formula [Cairns, 2004].

\subsubsection{Numerical Example for the Heston's Model}
For a better understanding of swaps pricing under the Heston's model, we provide following numerical example. \

We choose the value evaluated in the empirical study in section 2.2, see Table 1,  for the Heston's model parameters, which are evaluated through Markov chain Monte Carlo algorithm based on historical data of the S\&P 500 index over the period from January 13, 2015 to January 13, 2017 (One can refer to section 2.2.1 for details). The estimation of parameters in Heston's model are
\[
r=-0.0018, \kappa=0.8519, \theta=0.1574, \sigma=0.2403, \rho=-0.8740, V_0=0.0093.
\]
Therefore, the fair continuous variance strike of a S\&P 500 variance swap with one year maturity is
\begin{align}
    \nonumber K_{var}^*&=\theta+\frac{V_0-\theta}{\kappa T}(1-e^{-\kappa T})\\
    \nonumber &=0.1574+\frac{0.0093-0.1574}{0.8519\cdot 1 }(1-e^{-0.8519\cdot 1})\\
    &\approx 0.0577, \label{sh1}
\end{align}

the related fair continuous volatility strike derived from the convexity correction formula is
\begin{align}
  \nonumber K_{vol}^* &\approx \sqrt{0.0577}-\frac{\frac{0.2403^2e^{-2\cdot 0.8519}}{2\cdot 0.8519^3}\big[(0.0093-0.1574)(2e^{2\cdot 0.8519}-4e^{0.8519}\cdot0.8519-2)\big]}{8\big(0.0577\big)^{\frac{3}{2}}}\\
 \nonumber  &~~~-\frac{\frac{0.2403^2e^{-2\cdot 0.8519}}{2\cdot 0.8519^3}\big[0.1574(2e^{0.8519}\cdot 0.8519-3e^{2\cdot 0.8519}+4e^{0.8519}-1)\big]}{8\big(0.0577\big)^{\frac{3}{2}}}\\
 &\approx 0.3012.  \label{sh2}
\end{align}

Now we use the second approach -- the closed-form solution developed from Laplace transform to evaluate the related continuous volatility strike. From Theorem 2, we have that
\begin{align}
\gamma(s)&=\sqrt{\kappa^2+\frac{2\sigma^2s}{T}}=\sqrt{0.8519^2+2\cdot 0.2403^2s},\nonumber\\
     A(T,s)&=\frac{2\kappa\theta}{\sigma^2}\log\bigg(\frac{2\gamma(s)e^{\frac{T(\gamma(s)+\kappa)}{2}}}{(\gamma(s)+\kappa)(e^{T\gamma(s)}-1)+2\gamma(s)}\bigg)\nonumber\\
     &=\frac{2\cdot 0.8519\cdot0.1574}{0.2403^2}\log\bigg(\frac{2\gamma(s)e^{\frac{\gamma(s)+0.8519}{2}}}{(\gamma(s)+0.8519)(e^{\gamma(s)}-1)+2\gamma(s)}\bigg),\nonumber\\
     B(T,s)&=\frac{2s(e^{T\gamma(s)}-1)}{T(\gamma(s)+\kappa)(e^{T\gamma(s)}-1)+2T\gamma(s)}\nonumber\\
&=\frac{2s(e^{\gamma(s)}-1)}{(\gamma(s)+0.8519)(e^{\gamma(s)}-1)+2\gamma(s)},\nonumber\\
E(e^{-sRV_c(0,T)})&=\exp\big(A(T,s)-B(T,s)V_0\big)\nonumber\\
&=\exp\big(A(T,s)-B(T,s)\cdot 0.0093\big),\nonumber
\end{align}
thus the fair continuous volatility strike evaluated through the Laplace transform method is
\begin{align}
     K_{vol}^*&=\frac{1}{2\sqrt{\pi}}\int_0^\infty\frac{1-E(e^{-sRV_c(0,T)})}{s^{\frac{3}{2}}}ds\nonumber\\
     &\approx 0.1202 .\nonumber
\end{align}

\subsection{Merton Jump Diffusion Model}
In this section we consider the dynamic of underlying asset prices follow the Merton jump-diffusion model 
\begin{equation}
    \frac{dS_t}{S_t^-}=(r-\lambda m)dt+\sigma dW_t+dJ_t \label{m1}
\end{equation}
where $J_t=\sum_{i=1}^{N_t}(Y_i-1)$ is a compound Poisson process with intensity $\lambda$. $Y_i\sim Log-Normal(a,b^2)$ represent the jump size of price, and $E(Y_i-1)=m$ while the parameters are related by the equation $e^{a+\frac{1}{2}b^2}=m+1$. Moreover, when the jumps occur at time $\tau_i$, we have $S(\tau_i^+)=S(\tau_i^-)Y_i$ \cite{mark}.\

For a asset which can be modeled by (\ref{m1}), the variance of price comes from two parts: the diffusion of price process and jumps in price. Thus the continuous realized variance over $[0,T]$ in Merton jump diffusion model can be expressed as
\begin{equation}
    RV_c(0,T)=\frac{1}{T}\int_0^T\sigma^2dt+\frac{1}{T}\bigg(\sum_{i=1}^{N(T)}(\ln(Y_i))^2\bigg), \label{m2}
\end{equation}
and the fair continuous variance strike is
\begin{align}
    K_{var}^*&=E\big(RV_c(0,T)\big)=\frac{1}{T}\int_0^T\sigma^2dt+\frac{1}{T}E\bigg(\sum_{i=1}^{N(T)}(\ln(Y_i))^2\bigg)=\sigma^2+\lambda(a^2+b^2), \label{m3}
\end{align}
which depends on the volatility parameter $\sigma$ as well as the distribution of jump size.\

To evaluate the continuous volatility strike in the Merton jump diffusion model, we can either use convexity correction method or Laplace transform.
\begin{theorem}
From the convexity correction formula (\ref{ag2}) we have following approximation for fair continuous volatility strike in Merton jump diffusion model
\begin{align}
    K_{vol}^*&\approx \sqrt{\sigma^2+\lambda(a^2+b^2)}-\frac{\lambda(a^4+6a^2b^2+3b^4)}{8T\big(\sigma^2+\lambda(a^2+b^2)\big)^{\frac{3}{2}}}.
\end{align}
\end{theorem}

\begin{proof}
Since
\begin{align}
Var\big(RV_c(0,T)\big)&=Var\bigg(\frac{1}{T}\int_0^T\sigma^2dt+\frac{1}{T}\big(\sum_{i=1}^{N(T)}(\ln(Y_i))^2\big)\bigg)\nonumber\\
&=Var\bigg(\frac{1}{T}\big(\sum_{i=1}^{N(T)}(\ln(Y_i))^2\big)\bigg)\nonumber\\
&=\frac{1}{T^2}\cdot \lambda T\cdot E((\ln(Y_i))^4)\nonumber\\
&=\frac{\lambda(a^4+6a^2b^2+3b^4)}{T}, \label{what}
\end{align}
 substitute (\ref{m3}) and  (\ref{what}) into convexity correction formula, we have that
\begin{align}
    K_{vol}^*&\approx \sqrt{E\big(RV_c(0,T)\big)}-\frac{Var\big(RV_c(0,T)\big)}{8\big(E(RV_c(0,T))\big)^{\frac{3}{2}}}\nonumber\\
    &=\sqrt{\sigma^2+\lambda(a^2+b^2)}-\frac{\lambda(a^4+6a^2b^2+3b^4)}{8T\big(\sigma^2+\lambda(a^2+b^2)\big)^{\frac{3}{2}}}.\label{h330}
\end{align}
\end{proof}

\begin{theorem}
By applying Laplace transform method, we have following evaluation for the fair continuous volatility strike in Merton jump diffusion model
\begin{align}
     K_{vol}^*=E(\sqrt{RV_c(0,T)})=\frac{1}{2\sqrt{\pi}}\int_0^\infty\frac{1-E(e^{-sRV_c(0,T)})}{s^{\frac{3}{2}}}ds, \label{lp3}
\end{align}
where
\begin{equation*}
    E(e^{-sRV_c(0,T)})=\exp\bigg(-s\sigma^2+\lambda T\big(\frac{\exp(\frac{-sa^2}{T+2sb^2})}{\sqrt{1+\frac{2sb^2}{T}}}-1\big)\bigg) .
\end{equation*}
\end{theorem}

\begin{proof}
Since 
\begin{equation*}
    RV_c(0,T)=\sigma^2+\frac{1}{T}\bigg(\sum_{i=1}^{N(T)}(\ln(Y_i))^2\bigg), 
\end{equation*}
and 
\[
\ln(Y_i)\sim N(a,b^2),
\]
we have that
\begin{align}
     &E\bigg(\exp\big(-s(\sigma^2+\frac{1}{T}\sum_{i=1}^{N(T)}(\ln(Y_i))^2)\big)\bigg)\nonumber\\
     &=e^{-s\sigma^2}E\bigg(E\big(\exp\{-\frac{s}{T}\sum_{i=1}^{N(T)}(\ln(Y_i))^2\}|N(T)=n\big)\bigg)\nonumber\\
     &=e^{-s\sigma^2}\sum_{n=1}^{\infty}\frac{(\lambda T)^ne^{-\lambda T}}{n!}E\big(\exp\{-\frac{s}{T}\sum_{i=1}^{N(T)}(\ln(Y_i))^2\}\big)\nonumber\\
     &=\exp\bigg(-s\sigma^2+\lambda T\big(\frac{\exp(\frac{-sa^2}{T+2sb^2})}{\sqrt{1+\frac{2sb^2}{T}}}-1\big)\bigg).\nonumber
\end{align}
\end{proof}
See [Broadie and Jain, 2008, Prop. 3.1, page 771 and Prop. 5.1, page 774] for more details.

\subsubsection{Numerical Example for the Merton's Model}
Now we provide a numerical example for Merton's model. The parameters are evaluated through MCMC algorithm based on historical data of the S\&P 500 index over the period from January 13, 2015 to January 13, 2017. Let
\[
\lambda=0.0038, a=-0.0001, b^2=0.05, r=-0.0044, \sigma=0.1.
\]
Then, the fair continuous variance strike for a variance swap with maturity of one year is
\[
 K_{var}^*=\sigma^2+\lambda(a^2+b^2)\approx 0.0102.
 \]
 Using Theorem 3, we have following evaluation for fair continuous volatility strike
 \[
   K_{vol}^*\approx \sqrt{\sigma^2+\lambda(a^2+b^2)}-\frac{\lambda(a^4+6a^2b^2+3b^4)}{8T\big(\sigma^2+\lambda(a^2+b^2)\big)^{\frac{3}{2}}}\approx 0.097,
   \]
and the volatility strike evaluated from Laplace transform method is 
\begin{align}
 K_{vol}^*&=E(\sqrt{RV_c(0,T)})=\frac{1}{2\sqrt{\pi}}\int_0^\infty\frac{1-E(e^{-sRV_c(0,T)})}{s^{\frac{3}{2}}}ds\nonumber\\
&=\frac{1}{2\sqrt{\pi}}\int_0^\infty\frac{1-\exp\bigg(-s\sigma^2+\lambda T\big(\frac{\exp(\frac{-sa^2}{T+2sb^2})}{\sqrt{1+\frac{2sb^2}{T}}}-1\big)\bigg)}{s^{\frac{3}{2}}}ds\approx 0.0246    .\nonumber
\end{align}

\subsection{Bates Jump Diffusion Model}
Heston's and Merton's models are combined by Bates[1996] in 1996, who proposed the stochastic volatility with jumps model as following
\begin{align}
\nonumber &\frac{dS_t}{S_t^-}=(r-\lambda m)dt+\sqrt{V_t}dW_t^1+dJ_t, \\
 &dV_t=\kappa(\theta-V_t)dt+\sigma\sqrt{V_t}dW_t^2, \label{b1}
\end{align} 
where the meanings of parameters are same as in Heston's stochastic volatility model (\ref{h1}) and $J_t$ is a compound Poisson process with the same properties as in Merton jump diffusion model (\ref{m1}). Moreover, we assume that the jump process and Brownian motions are independent.\

Similar to the Merton's jump diffusion model (\ref{m1}), the variance of price comes from the diffusion of price process and jumps in price. Thus the continuous realized variance over $[0,T]$ in Bates jump diffusion model is
\begin{equation}
    RV_c(0,T)=\frac{1}{T}\int_0^TV_tdt+\frac{1}{T}\bigg(\sum_{i=1}^{N(T)}(\ln(Y_i))^2\bigg), \label{b2}
\end{equation}
and the fair continuous variance strike is
\begin{align}
 \nonumber   K_{var}^*&=E\big(RV_c(0,T)\big)=E\big(\frac{1}{T}\int_0^TV_tdt\big)+\frac{1}{T}E\bigg(\sum_{i=1}^{N(T)}(\ln(Y_i))^2\bigg)\\
 &=\theta+\frac{V_0-\theta}{\kappa T}(1-e^{-\kappa T})+\lambda(a^2+b^2). \label{b3}
\end{align}

Now we use both the convexity correction method and the Laplace transform method to evaluate the continuous volatility strike in the Bates jump diffusion model.

\begin{theorem}
From the convexity correction formula (\ref{ag2}) we have following approximation for fair continuous volatility strike under Bates jump diffusion model
\begin{align}
    K_{vol}^*&\approx \sqrt{E\big(RV_c(0,T)\big)}-\frac{Var\big(\frac{1}{T}\int_0^TV_tdt\big)}{8\big(E(RV_c(0,T))\big)^{\frac{3}{2}}}-\frac{Var\bigg(\frac{1}{T}\big(\sum_{i=1}^{N(T)}(\ln(Y_i))^2\big)\bigg)}{8\big(E(RV_c(0,T))\big)^{\frac{3}{2}}}.\label{b4}
\end{align}
where $E(RV_c(0,T))$ and $Var\big(\frac{1}{T}\int_0^TV_tdt\big)$ are given by (\ref{b3}) and (\ref{h23}) respectively, and 
\[
Var\bigg(\frac{1}{T}\big(\sum_{i=1}^{N(T)}(\ln(Y_i))^2\big)\bigg)=\frac{\lambda}{T}\cdot(a^4+6a^2b^2+3b^4).
\]
\end{theorem}

\begin{proof}
Since 
\[
 RV_c(0,T)=\frac{1}{T}\int_0^TV_tdt+\frac{1}{T}\bigg(\sum_{i=1}^{N(T)}(\ln(Y_i))^2\bigg),
 \]
 by using the convexity correction formula, we have that
\begin{align}
    K_{vol}^*&\approx \sqrt{E\big(RV_c(0,T)\big)}-\frac{Var\big(RV_c(0,T)\big)}{8\big(E(RV_c(0,T))\big)^{\frac{3}{2}}}\nonumber\\
    &=\sqrt{E\big(RV_c(0,T)\big)}-\frac{Var\big(\frac{1}{T}\int_0^TV_tdt\big)}{8\big(E(RV_c(0,T))\big)^{\frac{3}{2}}}-\frac{Var\bigg(\frac{1}{T}\big(\sum_{i=1}^{N(T)}(\ln(Y_i))^2\big)\bigg)}{8\big(E(RV_c(0,T))\big)^{\frac{3}{2}}}, \nonumber
\end{align}
and 
\begin{align} 
Var\bigg(\frac{1}{T}\big(\sum_{i=1}^{N(T)}(\ln(Y_i))^2\big)\bigg)&=\frac{1}{T^2}\cdot \lambda T\cdot E((\ln(Y_i))^4)\nonumber\\
&=\frac{\lambda(a^4+6a^2b^2+3b^4)}{T}. \nonumber
\end{align}
\end{proof}

\begin{theorem}
By applying Laplace transform method, we have following evaluation for the fair continuous volatility strike of Bates jump diffusion model
\begin{align}
     K_{vol}^*=E(\sqrt{RV_c(0,T)})=\frac{1}{2\sqrt{\pi}}\int_0^\infty\frac{1-E(e^{-sRV_c(0,T)})}{s^{\frac{3}{2}}}ds, \label{b5}
\end{align}
where
\begin{align} 
\nonumber E(e^{-sRV_c(0,T)})&=E(e^{-\frac{s}{T}\int_0^T V_t dt})\cdot E(e^{-\frac{s}{T}\sum_{i=1}^{N(T)}(\ln Y_i)^2})\\
\nonumber &=\exp\bigg(A(T,s)-B(T,s)V_0+\lambda T\big(\frac{\exp(\frac{-sa^2}{T+2sb^2})}{\sqrt{1+\frac{2sb^2}{T}}}-1\big)\bigg).
\end{align}  
$A(T,s)$ ans $B(T,s)$ are given by (\ref{ht1}) .
\end{theorem}
See [Broadie and Jain, 2008, Prop. 5.1, page 774] for more details.

\subsubsection{Numerical Example for the Bates' Model}
Now we provide a numerical example for Bates' model. The parameters are evaluated through MCMC algorithm based on historical data of the S\&P 500 index over the period from January 13, 2015 to January 13, 2017. Let
$r=-0.0044, \kappa=0.8269, \theta=0.1793, \sigma=0.2916, \rho=-0.8734, \lambda=0.0038, a=-0.0001, b^2=0.05, V_0=0.0103.$
Then, the fair continuous variance strike for a variance swap with maturity of one year is
\[
K_{var}^*=\theta+\frac{V_0-\theta}{\kappa T}(1-e^{-\kappa T})+\lambda(a^2+b^2)\approx 0.0645.
 \]
The fair volatility strike evaluated from convexity correction formula is 
\begin{align}
K_{vol}^*&=\sqrt{E\big(RV_c(0,T)\big)}-\frac{Var\big(\frac{1}{T}\int_0^TV_tdt\big)}{8\big(E(RV_c(0,T))\big)^{\frac{3}{2}}}-\frac{Var\bigg(\frac{1}{T}\big(\sum_{i=1}^{N(T)}(\ln(Y_i))^2\big)\bigg)}{8\big(E(RV_c(0,T))\big)^{\frac{3}{2}}}\nonumber\\
&\approx\sqrt{0.0645}-\frac{-0.0119}{8\big(0.0645\big)^{\frac{3}{2}}}-\frac{0.00003}{8\big(0.0645\big)^{\frac{3}{2}}}\approx 0.3445.   \nonumber
    \end{align}

The fair volatility strike evaluated from Laplace transform method given in Theorem 9 is
\begin{align}
K_{vol}^*&=\frac{1}{2\sqrt{\pi}}\int_0^\infty\frac{1-E(e^{-sRV_c(0,T)})}{s^{\frac{3}{2}}}ds\nonumber\\
&=\frac{1}{2\sqrt{\pi}}\int_0^\infty\frac{1-\exp\big(A(T,s)-B(T,s)V_0+\lambda T\big(\frac{\exp(\frac{-sa^2}{T+2sb^2})}{\sqrt{1+\frac{2sb^2}{T}}}-1\big)\big)}{s^{\frac{3}{2}}}ds    \nonumber\\
& \approx 0.1312 . \nonumber
\end{align}

\subsection{L\'evy Based Heston Model}
\subsubsection{$\alpha$-stable distributions and L\'evy processes}

In probability theory, a distribution is said to be stable if a linear combination of two independent copies of a random sample has the same distribution, up to location and scale parameters. Specifically, the characteristic function of symmetric $\alpha$-stable distributed random variables has following form
\[
\phi(u)=e^{i\delta u-\sigma|u|^\alpha},
\]
where $\alpha\in(0,2]$ is the characteristic exponent(stability parameter) which determines the shape of the distribution, $\delta\in(-\infty,\infty)$ is the location parameter and $\sigma\in(0,\infty)$ is the dispersion, which measures the width of distribution. For $0<\alpha\le1$, $\delta$ is the median, while for $1<\alpha\le2$, $\delta$ is the mean. A symmetric $\alpha$-stable distribution is called standard if $\delta=0$ and $\sigma=1$. For more details about symmetric $\alpha$-stable distribution, one can refer to \cite{swishchuk1}.\

However, there is no closed form expression exists for general $\alpha$-stable distribution other than the L\'evy ($\alpha=1/2$), the Cauchy ($\alpha=1$) and the Gaussian ($\alpha=2$) distributions. Also, only moments of order less than $\alpha$ exist for the non-Gaussian family of $\alpha$-stable distribution. The fractional lower order moments with $\delta=0$ are given by
\[
     E|X|^p=D(p,\alpha)\sigma^{p/\alpha}~~for~0<p<\alpha
\]
where
\[
D(p,\alpha)=\frac{2^p\Gamma(\frac{p+1}{2})\Gamma(1-\frac{p}{\alpha})}{\alpha\sqrt{\pi}\Gamma(1-\frac{p}{2})}
\]
and $\Gamma(\cdot)$ is the Gamma distribution.\

One important characteristic of symmetrical $\alpha$-stable distribution is that the smaller $\alpha$ is, the heavier the tails of the $\alpha$-stable density. The heavy tail characteristic makes the distribution appropriate for modeling noise which is impulsive in nature, for example, electricity prices or volatility (See [Swishchuk, 2009]).

{\bf Definition 2.1.2} Let $\alpha\in(0,2]$, an $\alpha$-stable L\'evy process $L_t$ is a process such that $L_1$ has a strictly $\alpha$-stable distribution($i.e.$, $L_1\equiv S_\alpha (\sigma,\beta,\delta)$ for some $\alpha\in(0,2]\setminus\{1\},  \sigma\in\mathbb{R}_+, \beta\in[-1,1], \delta=0$ or $\alpha=1, \sigma\in\mathbb{R}_+, \beta=0, \delta\in\mathbb{R}$). We call $L_t$ is a symmetric $\alpha$-stable L\'evy process if the distribution of $L_1$ is symmetric $\alpha$-stable ($i.e.$, $L_1\equiv S_\alpha(\sigma,0,0)$ for some $\alpha\in(0,2], \sigma\in\mathbb{R}_+$). $L_t$ is $(T_t)_{t\in\mathbb{R}_+}$-adapted if $L_t$ is a constant on $[T_{t-},T_{t+}]$ for any $t\in\mathbb{R}_+.$ \

The $\alpha$-stable L\'evy processes are the only self-similar L\'evy processes such that $L(at)\stackrel{\text{Law}}{=}a^{1/\alpha}L(t), a\ge 0$. They are either Brownian motion or pure jump. For $1<\alpha<2$, we have $E(L_t)=\delta t$ where $\delta$ is the location parameter of the $\alpha$-stable distribution. For more details about properties of $\alpha$-stable L\'evy processes, one can refer to [Swishchuk, 2009].\

\subsubsection{Change of Time Method for the Stochastic Differential Equations Driven by L\'evy Processes}
Let $L_{a.s.}^\alpha$ denotes the family of all real measurable $\mathcal{F}_t$-adapted processes $a(t,\omega)$ on $\Omega\times[0,+\infty)$, such that for every $T>0$,
\[
\int_0^T|a(t,\omega)|^\alpha dt<+\infty~~a.s..
\]
Now we consider stochastic differential equations that have following form
\[
dX(t)=a(t,X(t-))dL(t),
\]
where $L(t)$ is an $\alpha$-stable L\'evy process.\

\begin{theorem}
Let $a\in L_{a.s.}^\alpha$ such that $T(u):=\int_0^u |a|^\alpha dt\to +\infty$ a.s. as $u\to +\infty$. If $\hat{T}(t):=inf\{u:T(u)>t\}$ and $\hat{\mathcal{F}}_t=\mathcal{F}_{\hat{T}(t)}$, then the time-changed stochastic integral $\hat{L}(t)=\int_0^{\hat{T}(t)}a dL(t)$ is an $\hat{\mathcal{F}}_t$ $\alpha$-stable L\'evy process, where $L(t)$ is $\mathcal{F}_t$-adapted $\alpha$-stable L\'evy process. Consequently, for each $t>0$, $\int_0^ta dL=\hat{L}(T(t))$ a.s., i.e., the stochastic integral with respect to a $\alpha$-stable L\'evy process is nothing but another $\alpha$-stable L\'evy process with randomly changed time scale.
\end{theorem}\
See [Rosinski and Woyczinski, 1986] for more details. 

\subsubsection{Variance Swaps for the L\'evy-based Heston Model}

Assume the price and variance of underlying asset satisfy following model
\begin{align}
\nonumber&dS_t=r S_tdt+\sqrt{V_t}S_tdW_t \\ 
&dV_t=\kappa(\theta-V_t)dt+\sigma\sqrt{V_t}dL_t, \label{kkk1}
\end{align} 
where parameters have same meanings as in (\ref{h1}) while $W_t$ and $L_t$ are independent Brownian motion and $\alpha$-stable L\'evy process with $\alpha\in(0,2]$. By using the same method as in Section 5.1 and the change of time method [Swishchuk, 2009], we have following solution for the second SDE 
\begin{equation}
    V_t=\theta+e^{-\kappa t}\big(V_0-\theta+\hat{L}(\hat{T}_t)\big) \label{hh0}
\end{equation}
where 
\[
\hat{L}(\hat{T}_t)=\int_0^t \sigma e^{\kappa s}\sqrt{V_s}dL_s
\]
and 
\[
\hat{T}_t=\sigma^\alpha\int_0^t [e^{\kappa\hat{T}_s}(V_0-\theta+\hat{L}({\hat{T}_s}))+\theta e^{2\kappa\hat{T}_s}]^{\alpha/2}ds.
\]

Thus the fair continuous variance strike $K_{var}^*=E\big(RV_c(0,T)\big)$ is given by
\begin{align}
  K_{var}^* &=  E\big(\frac{1}{T}\int_0^TV_tdt\big)=E\bigg(\frac{1}{T}\int_0^T\theta+e^{-\kappa t}\big(V_0-\theta+\hat{L}(\hat{T}_t)\big)dt\bigg)\nonumber\\
    (Fubini's~Theorem)&=\frac{1}{T}\int_0^T\theta+e^{-\kappa t}\big(V_0-\theta+E\big(\hat{L}(\hat{T}_t)\big)\big)dt\nonumber\\
    &=\theta+\frac{(1-e^{-\kappa T})(\kappa V_0-\kappa\theta+\delta)}{\kappa^2 T}-\frac{\delta e^{-\kappa T}}{\kappa}.
\end{align}

However, only moments of order less than $\alpha$ exist for the non-Gaussian family of $\alpha$-stable distribution, which means we are not able to evaluate the variance of the realized variance $RV_c(0,T)$. Therefore, the convexity correction method are not able to be used to find the continuous volatility strike under the L\'evy-based Heston model.

\subsubsection{Numerical Example for the L\'evy-based Heston model}
Assume the dynamics of an asset price can be modeled as in (\ref{kkk1}), where the driven L\'evy process is symmetric $\alpha$-stable ($i.e., \beta=\delta=0$),  $\mu=-0.0018, \kappa=0.8519, \theta=0.1574, \sigma=0.2403, \rho=-0.8740,$ and $V_0=0.0093$. Then, the fair continuous variance strike of a variance swap with maturity of one year is given by
\begin{align}
 K_{var}^* &=  \theta+\frac{(1-e^{-\kappa T})(\kappa V_0-\kappa\theta+\delta)}{\kappa^2 T}-\frac{\delta e^{-\kappa T}}{\kappa}\nonumber\\
&=0.1574+\frac{(1-e^{-0.8519})(0.8519\cdot0.0093-0.8519\cdot0.1574)}{0.8519^2}\nonumber\\
&\approx 0.0577, \nonumber
\end{align}
which is same as the variance strike in numerical example of Heston's model in section 1.1.4.

\section{VIX Futures Pricing}

In this section, we will consider a highly traded volatility derivative -- the VIX future. We will price the VIX future under Heston's and Bates' models, and evaluate the pricing performance of different models with different approaches by comparing the estimated future prices with the market future prices.


\subsection{VIX Futures}
The Volatility Index (VIX) introduced by the Chicago Board Options Exchange (CBOE) in 1993 has been considered as a key measure of the stock market volatility. The original CBOE Volatility Index was designed to measure the market's expectation of 30-day implied volatility by at-the-money S\&P 100 Index option prices. In 2003, CBOE together with Goldman Sachs, updated the VIX to reflect a new way to measure expected volatility, which is based on the S\&P 500 Index (SPX) and estimates expected volatility by averaging the weighted prices of SPX puts and calls over a wide range of strike prices. CBOE introduced the first exchange-traded VIX futures contract on March 24, 2004 and launched VIX options after two years. The trading in VIX options and futures are very active and has grown to over 800,000 contracts per day in just 10 years since the launch. See [CBOE, 2014]\

As described in the CBOE white paper [CBOE, 2014], the generalized formula used in the VIX calculation is
\begin{equation}
\text{VIX}_t^2=\bigg(\frac{2}{\tau}\sum_i \frac{\Delta K_i}{K_i^2} e^{r\tau} Q(K_i)-\frac{1}{\tau}(\frac{F}{K_0}-1)^2\bigg)\times 100^2,
\end{equation} \label{ns1}
where $\tau=\frac{30}{365}$, $K_i$ is the strike price of the $i$th out-of-money option in the calculation, $F$ is the forward index level at time $t$, $Q(K_i)$ denotes the mid-quote price of the out-of-money options at strike $K_i$ at time $t$, $K_0$ is the first strike below the forward index level, and $r$ is the risk-free rate with maturity $\tau$. Mathematically, (\ref{ns1}) can be recognized as a simple discretization of the forward integral over $[t, t+\tau]$ [Lin, 2007], $i.e.$,
\begin{equation}
\text{VIX}_t^2=\bigg(\frac{\xi_1}{\tau}E_t\big(\int_t^{t+\tau}V_s ds\big)+\xi_2\bigg)\times 100^2, \label{ns2}
\end{equation} 
where $V_t$ is the instantaneous variance, $E(X)$ is the expectation under the risk-neutral probability measure and $E_t(X):=E(X|\mathcal{F}_t)$, $\xi_1$ and $\xi_2$ are coefficients determined by the price dynamics (See Appendix A in [Lin, 2007] for more details about the coefficient).\

The expression of the VIX squared can also be given in terms of risk-neutral expectation of the log contract 
[Zhu and Lian, 2011]
\begin{equation}
\text{VIX}_t^2=-\frac{2}{\tau} E_t \big(\ln(\frac{S_{t+\tau}}{S_t e^{r\tau}})\big)\times 100^2. \label{zhuu1}
\end{equation}

Carr and Wu [2006] showed that the price of a VIX future is a martingale under the risk-neutral measure, and the value of a VIX future contract with maturity $T$ is
\begin{equation}
F(T)=E(\text{VIX}_T).
\end{equation} \label{ns3}

\subsubsection{VIX Futures Pricing under the Heston's Model}
Assume the dynamics of S\&P 500 Index can be approximated by Heston's stochastic volatility model as in (\ref{h1}), in which we have that $\xi_1=1$ and $\xi_2=0$. Thus the VIX squared in this case is
\begin{align}
\nonumber\text{VIX}_t^2&=\frac{100^2}{\tau}E_t\big(\int_t^{t+\tau}V_s ds\big)\\
&=100^2\times \bigg(\theta+\frac{V_t-\theta}{\kappa\tau}(1-e^{-\kappa\tau})\bigg), \label{ns4}
\end{align} 
and the present value of a VIX future contract with maturity $T$ is
\begin{align} 
\nonumber F(T)&=E(\text{VIX}_T)\\
&=100\times E\bigg(\sqrt{\theta+\frac{V_T-\theta}{\kappa\tau}(1-e^{-\kappa\tau})}\bigg).  \label{ns5}
\end{align}

Now we use convexity correction formula and Laplace transform method to evaluate $F(T)$ separately.\

\begin{theorem}
Applying the convexity correction formula (\ref{ag2}) to (\ref{ns5}), we have following approximation for the value of VIX future contract,
\begin{align}
 \nonumber   F(T)&\approx100\times\bigg(\sqrt{\theta-\frac{\theta(1-e^{-\kappa\tau})}{\kappa\tau}+\frac{1-e^{-\kappa\tau}}{\kappa\tau}\cdot E(V_T)}\nonumber\\
 &-\frac{(1-e^{-\kappa\tau})^2Var(V_T)}{8\kappa^2\tau^2\big(\theta-\frac{\theta(1-e^{-\kappa\tau})}{\kappa\tau}+\frac{1-e^{-\kappa\tau}}{\kappa\tau}\cdot E(V_T)\big)^{3/2}}\bigg), 
\end{align} 
where 
\[
E(V_T)=\theta+(V_0-\theta) e^{-\kappa T}
\]
and 
\begin{equation}
Var(V_T)=\frac{\theta\sigma^2}{2\kappa}(1-e^{-2\kappa T})+\frac{(V_0-\theta)\sigma^2}{\kappa}(e^{-\kappa T}-e^{-2\kappa T}). \label{ns6}
\end{equation}
\end{theorem}\
 \begin{proof}
 
 By applying convexity correction formula, we have that
 \begin{align}
 \nonumber   F(T)&=100\times E\bigg(\sqrt{\theta+\frac{V_T-\theta}{\kappa\tau}(1-e^{-\kappa\tau})}\bigg) \nonumber\\
\nonumber&\approx 100\times\bigg(\sqrt{E\big(\theta+\frac{V_T-\theta}{\kappa\tau}(1-e^{-\kappa\tau})\big)}-\frac{Var\big(\theta+\frac{V_T-\theta}{\kappa\tau}(1-e^{-\kappa\tau})\big)}{8\big(E(\theta+\frac{V_T-\theta}{\kappa\tau}(1-e^{-\kappa\tau}))\big)^{\frac{3}{2}}}\bigg)\\
 &=100\times\bigg(\sqrt{\theta-\frac{\theta(1-e^{-\kappa\tau})}{\kappa\tau}+\frac{1-e^{-\kappa\tau}}{\kappa\tau}\cdot E(V_T)}\nonumber\\
 &-\frac{(1-e^{-\kappa\tau})^2Var(V_T)}{8\kappa^2\tau^2\big(\theta-\frac{\theta(1-e^{-\kappa\tau})}{\kappa\tau}+\frac{1-e^{-\kappa\tau}}{\kappa\tau}\cdot E(V_T)\big)^{3/2}}\bigg)\nonumber. 
\end{align} 
\end{proof}

Zhu and Lian [2011] consider a general model for the S\&P 500 which incorporates stochastic volatility and simultaneous jumps in both the asset price and the volatility process. They found the closed-form pricing formula for the exact price of a VIX future by solving a backward partial integro-differential equation (PIDE). Now we give following closed-form pricing formula for the Heston stochastic volatility model by modifying the result in [Zhu and Lian, 2011].
\begin{theorem}
Assume the dynamics of S$\&$P 500 Index is given by the Heston stochastic volatility model (\ref{h1}), the price of a VIX future with maturity $T$ is then 
\begin{equation}
F(T,V_0)=\frac{1}{2\sqrt{\pi}}\int_0^\infty \frac{1-e^{-100^2sB}f(-100^2sA;T,V_0)}{s^{3/2}}ds, \label{ns7}
\end{equation}
where 
\[
A=\frac{1-e^{-\kappa\tau}}{\kappa\tau},~~B=\theta(1-\frac{1-e^{-\kappa\tau}}{\kappa\tau}),
\]
and $f(\phi;T,V_0)$ is the moment generating function of the stochastic variable $V_T$, given by
\[
f(\phi;T,V_0)=e^{C(\phi,T)+D(\phi,T)V_0}
\]
with 
\[
C(\phi,T)=\frac{-2\kappa\theta}{\sigma^2}\cdot\ln\big(1+\frac{\sigma^2\phi}{2\kappa}(e^{-\kappa T}-1)\big)
\]
and
\[
D(\phi,T)=\frac{2\kappa\phi}{\sigma^2\phi+(2\kappa-\sigma^2\phi)e^{\kappa T}} .
\]
\end{theorem}
See [Zhu and Lain, 2011] for more details.

\subsubsection{VIX Futures Pricing under the Bates' Model}
In this section we derive the present value of VIX futures in Bates' model (\ref{b1}), where $\xi_1=1$ and $\xi_2=2\lambda(m-a)$. By the definition of VIX squared (\ref{ns2}), we have that
\begin{align}
\nonumber\text{VIX}_t^2&=\bigg(\frac{1}{\tau}E_t\big(\int_t^{t+\tau}V_s ds\big)+2\lambda(m-a)\bigg)\times 100^2\\
&=100^2\times \bigg(\theta+\frac{V_t-\theta}{\kappa\tau}(1-e^{-\kappa\tau})+2\lambda(m-a)\bigg), \label{ns8}
\end{align}
and the present value of a VIX future contract with maturity $T$ is
\begin{align} 
\nonumber F(T)&=E(\text{VIX}_T)\\
&=100\times E\bigg(\sqrt{\theta+\frac{V_T-\theta}{\kappa\tau}(1-e^{-\kappa\tau})+2\lambda(m-a)}\bigg).  \label{ns9}
\end{align}
Once again, we use convexity correction formula and Laplace transform method to evaluate $F(T)$ in (\ref{ns9}).\

\begin{theorem}
By using the convexity correction formula (\ref{ag2}), we have following approximation for the value of VIX future contract in Bates' model,
\begin{align}
 \nonumber   F(T)&\approx100\times\bigg(\sqrt{\theta-\frac{\theta(1-e^{-\kappa\tau})}{\kappa\tau}+\frac{1-e^{-\kappa\tau}}{\kappa\tau}\cdot E(V_T)+2\lambda(m-a)}\\
 &-\frac{(1-e^{-\kappa\tau})^2Var(V_T)}{8\kappa^2\tau^2\big(\theta-\frac{\theta(1-e^{-\kappa\tau})}{\kappa\tau}+\frac{1-e^{-\kappa\tau}}{\kappa\tau}\cdot E(V_T)+2\lambda(m-a)\big)^{3/2}}\bigg), \label{ns10}
\end{align} 
where $E(V_T)$ and $Var(V_T)$ are same as in (\ref{ns6}).
\end{theorem}
\begin{proof}
By using convexity correction formula, we have that
\begin{align}
 \nonumber   F(T)&=100\times E\bigg(\sqrt{\theta+\frac{V_T-\theta}{\kappa\tau}(1-e^{-\kappa\tau})+2\lambda(m-a)}\bigg) \\
\nonumber&\approx 100\times\bigg(\sqrt{E\big(\theta+\frac{V_T-\theta}{\kappa\tau}(1-e^{-\kappa\tau})+2\lambda(m-a)\big)}\nonumber\\
&-\frac{Var\big(\theta+\frac{V_T-\theta}{\kappa\tau}(1-e^{-\kappa\tau})+2\lambda(m-a)\big)}{8\big(E(\theta+\frac{V_T-\theta}{\kappa\tau}(1-e^{-\kappa\tau}))+2\lambda(m-a)\big)^{\frac{3}{2}}}\bigg)\\
 \nonumber&=100\times\bigg(\sqrt{\theta-\frac{\theta(1-e^{-\kappa\tau})}{\kappa\tau}+\frac{1-e^{-\kappa\tau}}{\kappa\tau}\cdot E(V_T)+2\lambda(m-a)}\\
 &-\frac{(1-e^{-\kappa\tau})^2Var(V_T)}{8\kappa^2\tau^2\big(\theta-\frac{\theta(1-e^{-\kappa\tau})}{\kappa\tau}+\frac{1-e^{-\kappa\tau}}{\kappa\tau}\cdot E(V_T)+2\lambda(m-a)\big)^{3/2}}\bigg) \nonumber. 
\end{align} 
\end{proof}

\begin{theorem}
Assume the dynamics of S$\&$P 500 Index is given by the Bates jump model, then the price of a VIX future with maturity $T$ is given by
\begin{equation}
F(T,V_0)=\frac{1}{2\sqrt{\pi}}\int_0^\infty \frac{1-e^{-100^2sB}f(-100^2sA;T,V_0)}{s^{3/2}}ds, \label{ns11}
\end{equation}
where 
\[
A=\frac{1-e^{-\kappa\tau}}{\kappa\tau},~~B=\theta(1-\frac{1-e^{-\kappa\tau}}{\kappa\tau})+2\lambda(m-a),
\]
and $f(\phi;T,V_0)$ is the moment generating function of the stochastic variable $V_T$, given by
\[
f(\phi;T,V_0)=e^{C(\phi,T)+D(\phi,T)V_0}
\]
with 
\[
C(\phi,T)=\frac{-2\kappa\theta}{\sigma^2}\cdot\ln\big(1+\frac{\sigma^2\phi}{2\kappa}(e^{-\kappa T}-1)\big)
\]
and
\[
D(\phi,T)=\frac{2\kappa\phi}{\sigma^2\phi+(2\kappa-\sigma^2\phi)e^{\kappa T}} .
\]
\end{theorem}
See [Zhu and Lian, 2011] for more details.

\subsection{Empirical Studies}
In this part we will use historical data of the S\&P 500 index and the pricing formulas derived in above section to price the VIX futures, and evaluate the pricing performance by comparing the estimated prices with market prices of VIX futures.

\subsubsection{Calibration}
It has been shown that the Markov chain Monte Carlo (MCMC) algorithm outperforms some other calibration methods in many ways. Its' advantages such as stability, computational efficiency, the ability of detecting jumps [Cape {\it et al.}, 2015] make it suitable for parameters estimation in our cases. In this section, we use MCMC algorithm to estimate the model parameters from the historical data of the S\&P 500 index over the period from January 13, 2015 to January 13, 2017.\

In our study, we use the method provided in [Cape {\it et al.}, 2015] and [Johannes and Polson, 2006], which use Gibbs sampler for parameter estimation and Metropolis-Hasting algorithm for simulating the variance process $V_t$. We implement the MCMC calibration by using the R package provided by the authors in [Cape {\it et al.}, 2015]. The calibration procedure was applied to Heston's and Bates' models respectively. The following were chosen as the prior  distribution parameters, 
\begin{flalign*}
r \sim N(0,1),\\
\kappa\sim N(0,1),\\
\theta \sim N(0,1),\\
\psi:=\rho\sigma \sim N(0,\frac{\Omega}{2}),\\
\Omega:=\sigma^2(1-\rho^2) \sim IG(2,\frac{1}{200}),\\
\lambda\sim \text{Beta}(2,40),\\
a\sim N(0,1),\\
b^2\sim IG(5.0,0.2).
\end{flalign*}
Initial values for the MCMC algorithm were chosen based off the observed data when possible or a random assignment when more educated estimates were not possible (see [Cape {\it et al.}, 2015]). As a result, the following initials were chosen:
\begin{align*}
r^{(0)}= 0.1,\\
\kappa^{(0)}=5,\\
\theta^{(0)} = 0.0225,\\
\Omega =0.02,\\
\psi^{(0)} \sim N\left(0,\frac{\Omega^{(0)}}{2}\right),
\end{align*}
\begin{align*}
\lambda^{(0)}\sim \text{Beta}(2,40),\\
a^{(0)}=0,\\
b^{2(0)}= 0.1.
\end{align*}

After our simulations, we discarded the first 3000 runs as `burn-in' period and used the last 8,000 iterations to estimate model parameters. Means of the draws from the posterior distributions of each parameter are reported as well as the standard deviation of the draws for the distribution. The algorithm was run 10 times, recording the parameter values after each run after which the means were calculated from the ten runs. Each run took about 20 minutes and was done completely in the statistical language of R, utilizing pre-defined routines for random number generation. Table \ref{table1} provides a summary of the results obtained from the MCMC simulations.
\begin{table}[h!]
\centering
  \begin{tabular}{ c  c  c }
    \hline
    Parameters~~~&Heston~~~ & ~Bates \\ \hline

     $r$  & $-0.0018$ & $-0.0044$ \\ 
       & $(0.0794)$ & $(0.0824)$ \\ 
      $\kappa$  & $0.8519$ & $0.8269$ \\ 
       & $(0.7590)$ & $(0.7239)$ \\ 
      $\theta$ & $0.1574$ & $0.1793$ \\ 
       & $(0.2939)$ & $(0.2959)$ \\ 
      $\sigma$  & $0.2403$ & $0.2916$ \\ 
       & $(0.0768)$ & $(0.0384)$ \\ 
      $\rho$  &$-0.8740$ & $-0.8734$ \\ 
       &$(0.0478)$ & $(0.0439)$ \\ 
      $\lambda$  &  & $0.0038$ \\ 
       & & $(0.0027)$ \\ 
      $a$  & & $-0.0001$ \\ 
      & & $(0.9985)$ \\ 
      $b^2$  &  & $0.0500$ \\ 
       & & $(0.0294)$ \\ \hline
  \end{tabular}
  \caption{Means and standard deviation of estimated parameters} \label{table1}
\end{table}\

Like other published results, e.g., [Cape {\it et al.}, 2015], [Zhu and Lian, 2011], there is a strong negative correlation between the instantaneous volatility and returns, and the correlation is even stronger than others that have been observed. The estimation of $\lambda$ indicate that the jump happens very infrequently, with roughly one jump observed per year. Although Cape {\it et al.} [2015] point out that the MCMC algorithm we used has the difficulty in detecting jumps during times of high volatility such as the late 2008, the S\&P 500 index is relatively stable during the period we choose.

\subsubsection{Comparative Studies in VIX Future Pricing Performance}
In this section, we use VIX futures market prices as the benchmark, and compare the pricing performance of Heston's and Bates' models under convexity correction approximation and closed-form solution pricing formula. By following the studies in [Habtemicael and SenGupta, 2017], we employ following measures of ``goodness of fit'' of the estimated VIX future prices: the absolute percentage error (APE), the average absolute error (AAE), the average relative percentage error (ARPE) the root-mean-square error (RMSE) and the residual standard error (RSE), which are given by
\[
\text{APE}=\frac{1}{\text{mean price}}\sum_{\text{data points}}\frac{|\text{market price - model price}|}{\text{data points}},
\]
\[
\text{AAE}=\sum_{\text{data points}}\frac{|\text{market price - model price}|}{\text{data points}},
\]
\[
\text{ARPE}=\frac{1}{\text{data points}}\sum_{\text{data points}}\frac{|\text{market price - model price}|}{\text{data points}},
\]
\[
\text{RMSE}=\sqrt{\sum_{\text{data points}}\frac{|\text{market price - model price}|}{\text{data points}}}, ~~~\text{RSE}=\sqrt{\frac{SSE}{n-k}},
\]
where SSE is the sum of square error, $n$ is the number of observations and $k$ is the number of parameters to be estimated. By using the estimated parameters reported in Table \ref{table1}, we compute the VIX futures prices with different maturities on Jan 13, 2017. The values of APE, AAE, ARPE, RMSE and RSE are tabulated in Table \ref{table2}.

\begin{sidewaystable}
\centering
  \begin{tabular}{ c  c  c  cc ccc cccc}
    \hline
    Pricing Errors & Models and Pricing Methods & All Futures & $T\le 30$ &   $30<T\le 90$ & $90<T$ \\ \hline
    APE       
    & Heston (Convex) & 0.1928 & 0.0497 & 0.1022 & 0.3114 \\ 
       & Heston (Closed-form) & 0.0774 & 0.0737 & 0.0473 & 0.0959 \\ 
      & Bates (Convex) & 0.1811 & 0.1143 & 0.0472 & 0.2873 \\  
   & Bates (Closed-form) & 0.0820 & 0.0258 & 0.0291 & 0.1383 \\ 
       &\\ 
      AAE        
      & Heston (Convex) & 3.0759 & 0.6531 & 1.5556 & 5.7046 \\ 
       & Heston (Closed-form) & 1.2353 & 0.9679 & 0.7201 & 1.7569 \\ 
      & Bates (Convex) & 2.8896 & 0.4289 & 0.7183 & 5.2627 \\  
   & Bates (Closed-form) & 1.3087 & 0.3384 & 0.4431 & 2.5326 \\ 
      & \\ 
     ARPE    
     & Heston (Convex) & 0.2197 & 0.1633 & 0.3889 & 0.9508 \\ 
       & Heston (Closed-form) & 0.0882 & 0.2420 & 0.1800 & 0.2928 \\ 
      & Bates (Convex) & 0.2064 & 0.3753 & 0.1796 & 0.8771 \\  
   & Bates (Closed-form) & 0.0935 & 0.0846 & 0.1108 & 0.4221 \\   
    & \\ 
     RMSE    
     & Heston (Convex) & 1.7538 & 0.8081 & 1.2472 & 2.3884 \\ 
       & Heston (Closed-form) & 1.1114 & 0.9838 & 0.8486 & 1.3255 \\ 
      & Bates (Convex) & 1.6999 & 0.6549 & 0.8475 & 2.2941 \\  
   & Bates (Closed-form) & 1.1440 & 0.5817 & 0.6657 & 1.5914 \\ 
      & \\ 
     RSE
     & Heston (Convex) & 6.0902 & 0.5559 & 1.3931 & 5.9026 \\ 
       & Heston (Closed-form) & 2.2036 & 0.7935 & 0.6515 & 1.9498 \\ 
      & Bates (Convex) & 5.7498 & 1.2680 & 0.7135 & 5.5626 \\  
   & Bates (Closed-form) & 2.7307 & 0.2866 & 0.4435 & 2.6792   \\ 
     \hline
  \end{tabular}
  \caption{The Test of Pricing Performance} \label{table2}
\end{sidewaystable}

From the Table \ref{table2}, we can draw some conclusions about the pricing performance. All these five different measures of pricing performance show that the VIX futures prices estimated from closed-form solutions are more accurate than those estimated from convexity correction approximation for both Heston's and Bates' models, generally. However, the convex correction method outperforms the closed-form solution method for the short-term futures in the Heston's model. Also, for short-term and medium-term futures, Bates' model with closed-form solution performs better than the other cases. \

\begin{figure}[h!]
\begin{center}
\includegraphics[scale=0.8]{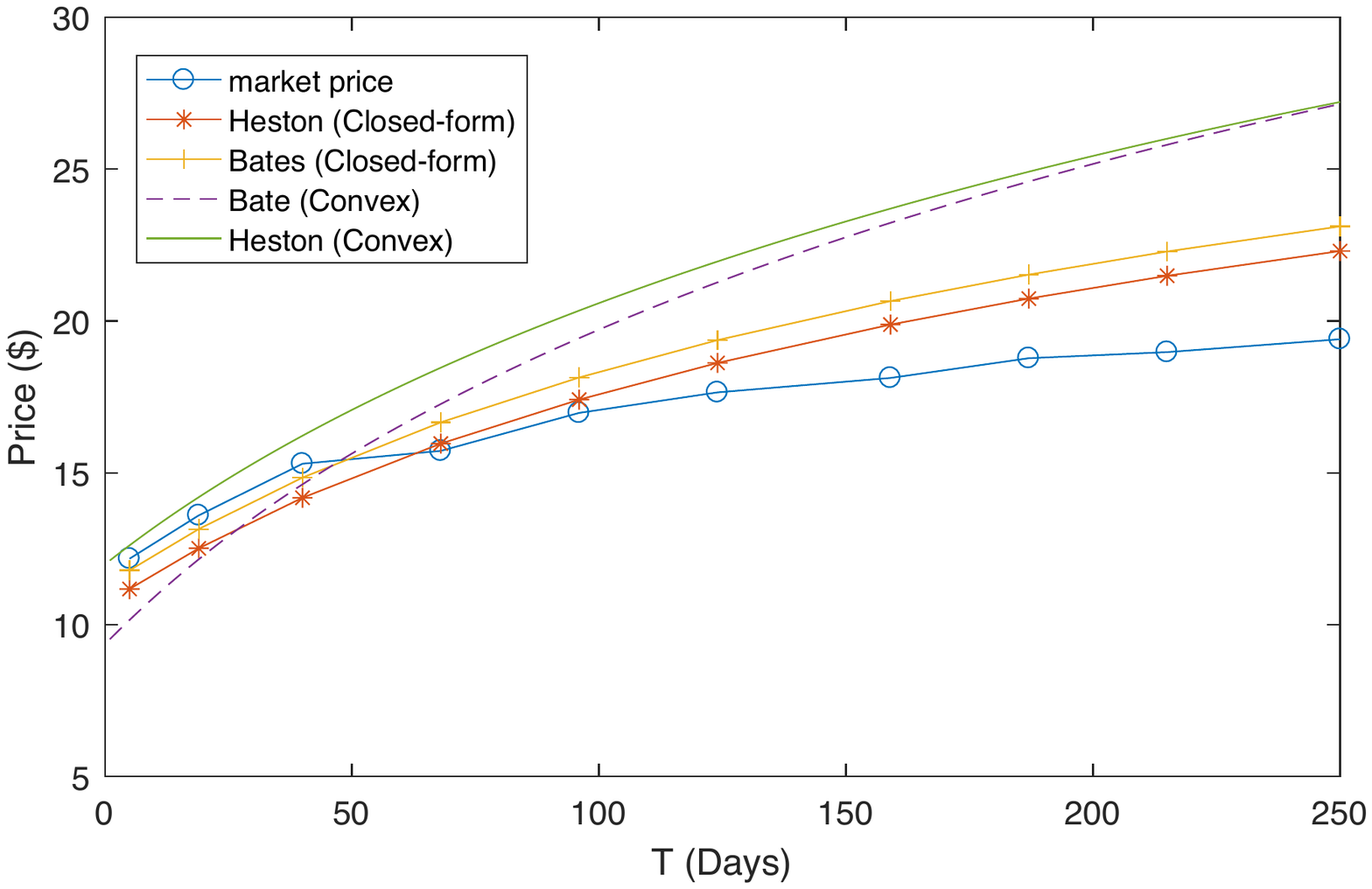}
\end{center}
\caption{Comparison of VIX Futures Market Price with Estimated Price}
\label{vixprice}
\end{figure}

To illustrate the pricing performance more clearly, we plot the market prices of VIX futures and estimated prices on the same graph in Figure \ref{vixprice}. It can be observed that the Heston model with convexity correction approximation always overvalue the futures; for the short-term VIX futures, the Bates model with closed-form solution provides the best estimation; for the VIX futures with medium to long term maturities, all of the pricing methods will over-price the futures. \

\begin{figure}[h!]
\begin{center}
\includegraphics[scale=0.8]{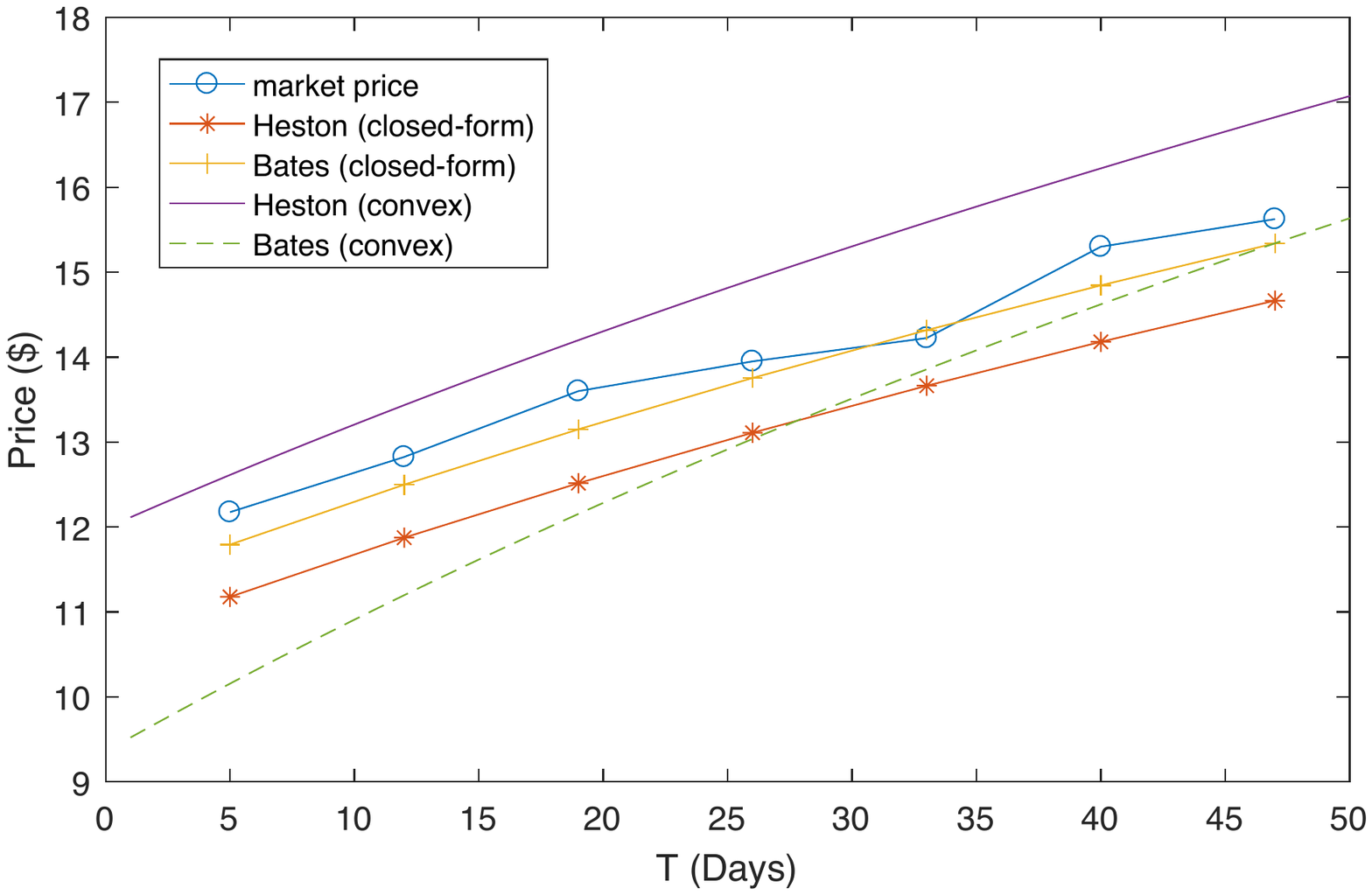}
\end{center}
\caption{Comparison of Short-term VIX Futures Market Price with Estimated Price}
\label{vixprice1}
\end{figure}

However, we think the trading volumes of VIX futures is a main reason for the poor pricing performance for long-term VIX futures. From Table \ref{table3} we can see that the trading of short-term VIX future (within one month) are very active, and the total trading volume decreases significantly as the time-to-expire increases. The data we downloaded from CBOE website shows that there is no VIX future with maturity longer than 250 days had been traded in the market. To some degree, this explain the deviation of the estimated price and the market price of long-term futures as in Figure \ref{vixprice}. The prices of long-term VIX futures with low trading volume are model-free and thus cannot reflect the market expectation on the underlying asset.

\begin{table}[h!]
\centering
  \begin{tabular}{ c  c }
    \hline
    Time to expire(days) & Total volume(contracts) \\ \hline
    5 & 110184\\
    33 &  113493 \\ 
    68 & 34580 \\
      96 & 12146 \\ 
      124 & 7351 \\  
     159 & 5007 \\   
     187 &      1815     \\
     215 &        343      \\
     250 &  0\\
      \hline
  \end{tabular}
  \caption{Total Trading Volumes of VIX Futures with Different Maturities on Jan 13, 2017} \label{table3}
\end{table}

For better comparison of pricing performance for short-term VIX futures, we pick more short-term VIX futures and repeat the procedure, and we get the pricing result as in Figure \ref{vixprice1}, which shows that Bates (closed-form) can provide relatively reliable estimation of the market price, and it is even more accurate when the time to expire is between 25 to 35 days. We want to point out that the conclusion would be more convincing if we took VIX futures on different date into account instead of only considering the pricing performance on Jan 13, 2017. One possible way to do that is collect the market data, sort all the observed futures pricing according to expiration,  group these futures by every 30 days to expiration, and then compute the average prices of each group [Zhu 2011]. Due to the computational complexity, in this thesis we just use the data on a specific date to provide an intuitively understanding, although the conclusions coincide with those in [Zhu, 2011].

\section{Conclusion} In this paper, we considered variance and volatility swaps pricing for different stochastic volatility models, such as Heston, Bates, Merton and L\'evy-based Heston models, and presented numerical results based on historical data of the $S\&P$ 500 Index, January 13, 2005-January 13, 2017. We also studied VIX futures pricing for the Heston and the Bates models, presented empirical studies for them, based on the above-mentioned  data,  and performed comparative studies.

\section*{Acknowledgements:} The authors wish to thank NSERC for continuing support.

\section{References}
\hspace{0.5cm}

Bates, D. (1996). Jump and Stochastic Volatility: Exchange Rate Processes Implicit in Deutsche Mark in Options,
{\it Review of Financial Studies}, 9, pp. 69-107.

Brockhaus, O. and Long, D. (2002). Volatility Swaps Made Simple, {\it  Risk}, 19(1), pp. 92-95.

Broadie, M. and A. Jain, A. (2008). The Effect of Jumps and Discrete Sampling on Volatility and Variance Swaps,
 {\it International Journal of Theoretical and Applied Finance}, Vol.11, No.8, pp. 761-797.

Cairns, A. (2004). {\it Interest Rate Models: An Introduction}.
(Princeton University Press, USA).

CBOE (2014). The CBOE Volatility Index - VIX. {\it White Paper}.\\
(http://www.cboe.com/micro/vix/vixwhite.pdf)

 Cape, J., Dearden, W., Gamber, W., Liebner, J., Lu, Q.  and Nguyenu, M. (2015). Estimating Heston's and Bates'€™ Models Parameters Using Markov Chain Monte Carlo Simulation,
 {\it Journal of Statistical Computation and Simulation},  Volume 85, Issue 11.

Carr, P and Wu, L. (2006).  A Tale of Two Indices,
{\it The Journal of Derivatives}, 13 (3).

Heston, S. (1993). A Closed-Form Solution for
Options with Stochastic Volatility with Applications to Bond and Currency Options,
{\it The Review of Financial Studies}, 6(2), pp. 327-343.

 Habtemicael, S. and SenGupta, I. (2016).
Pricing Variance and Volatility Swaps for Barndorff-Nielsen and Shephard Process Driven Financial Markets,
{\it €'International Journal of Financial Engineering}, Vol. 03, Issue 04.

Johannes, M. and  N. Polson, N. (2006).
MCMC Methods for Continuous-Time
Financial Econometrics. In: {\it Handbook in Financial Econometrics}, Vol. 2, Chapter 13, pp. 1-72. (Ed. Y. Ait-Sahalia and L.P. Hansen).

Lin, Y. (2007). Pricing VIX Futures: Evidence from Integrated Physical and Risk-neutral Probability Measures,
{\it  Journal of Futures Markets}, 27(12), pp. 1175 - 1217.

Rosinski, J and Woyczinski, W. On Ito Stochastic Integration With Respect To p-stable Motion: Inner Clock, Integrability of Sample Paths, Double and Multiple Integrals,
{\it Annals of Probability}, 14, pp. 271-286.

 Swishchuk, A. (2009).  Multi-Factor L\'evy Models for Pricing Financial and Energy Derivatives,
 {\it Canadian Applied Mathematics Quarterly}, Vol.17, No.4, Winter.

Zhu, S. and Lian, G. (2011). An Analytical Formula for VIX Futures and Its Applications,
 {\it Journal of Futures Markets}, 32(2), pp. 166 - 190.

\hspace{0.5cm}

\end{document}